\documentclass{scrartcl}

\usepackage[format=plain]{caption}
\usepackage{natbib}
\usepackage[dvips]{graphicx}
\usepackage{hyperref}
\usepackage[noblocks]{authblk}
\usepackage[utf8]{inputenc} 
\usepackage{microtype}

\usepackage{amsthm, amsmath, amsfonts, mathtools, amssymb} 
\usepackage{sgame} 

\usepackage{multirow,array}
\usepackage{tikz}
\usetikzlibrary{trees} 
\usetikzlibrary{calc}
\usetikzlibrary{matrix}
\usetikzlibrary{positioning}
\usetikzlibrary{3d}
\usetikzlibrary{fadings,decorations.pathreplacing}
\usetikzlibrary{arrows}
\usepackage[mathscr]{euscript}
\usepackage{enumitem}
\usepackage{bm}
\usepackage{color}
\usepackage{makecell}
\usepackage{multicol}
\usepackage{float}
\restylefloat{table}

\usepackage[capitalize]{cleveref}
\usepackage{times}

\def\LA{\rlap{\scalebox{1.2}[1.5]{\kern2pt\raisebox{-.5pt}{$\leftarrow$}}}}
\def\RA{\rlap{\scalebox{1.2}[1.5]{\kern2pt\raisebox{-.5pt}{$\rightarrow$}}}}

\newcommand*{\pup}{p_u^+}
\newcommand*{\pum}{p_u^-}
\newcommand*{\prp}{p_r^+}
\newcommand*{\prm}{p_r^-}
\newcommand*{\parp}{p_{ar}^+}
\newcommand*{\parm}{p_{ar}^-}

\newcommand*{\eup}{e_u^+}
\newcommand*{\eum}{e_u^-}
\newcommand*{\Grp}{G_r^+}
\newcommand*{\Grm}{G_r^-}
\newcommand*{\Garp}{G_{ar}^+}
\newcommand*{\Garm}{G_{ar}^-}

\def\SetR{\mathbb{R}}           

\tikzset{
solid node/.style={circle,draw,inner sep=1.5,fill=black},
hollow node/.style={circle,draw,inner sep=1.5}
}

\newtheorem{theorem}{Theorem}

\newtheorem{proposition}[theorem]{Proposition}
\newtheorem{corollary}[theorem]{Corollary}

\crefname{claim}{Claim}{Claims}

\crefname{statement}{Statement}{Statements}
\theoremstyle{definition}
\newtheorem{definition}[theorem]{Definition}

\crefname{assumption}{Assumption}{Assumptions}

\crefname{customassumption}{Assumption}{Assumptions}
\theoremstyle{remark}

\theoremstyle{definition}

\newtheorem{primary statistics}[theorem]{Primary Statistics}
\newtheorem{auxiliary statistics}[theorem]{Auxiliary Statistics}

\newenvironment{keywords}%
{\begin{abstract}\noindent}%
{\end{abstract}}

\begin{document}

\author{Tobias W{\"a}ngberg}
\author{Mikael B{\"o}{\"o}rs}
\affil{Link{\"o}ping University, 581 83 Link{\"o}ping, Sweden}
\author{Elliot Catt}
\author{Tom Everitt}
\author{Marcus Hutter}
\affil{Australian National University, Acton 2601, Australia}

\title{A Game-Theoretic Analysis of The Off-Switch Game}
\date{June 12, 2017}

\maketitle

\begin{abstract}
  \noindent
  \emph{Abstract:}
	The off-switch game is a game theoretic model of a highly
  intelligent robot interacting with a human.
  In the original paper by \citet{Hadfield-Menell2016},
  the analysis is not fully game-theoretic as the human is modelled
  as an irrational player, and the robot's best action is only
  calculated under unrealistic normality and soft-max assumptions.
  In this paper, we make the analysis fully game theoretic,
  by modelling the human as a rational player with a random utility
  function.
  As a consequence, we are able to easily calculate
  the robot's best action for arbitrary belief and irrationality
  assumptions.
\end{abstract}


\tableofcontents

\pagebreak
\section{Introduction}

Artificially intelligent systems are often created to satisfy some
goal. For example, \emph{Win a chess game} or \emph{Keep the house clean}.
Almost any goal can be formulated in terms of a reward or utility
function $U$ that maps states and actions to real numbers \citep{vonNeumann1947}.
This utility function may either be preprogrammed by the designers,
or learnt \citep{Dewey2011}.

A core problem in Artificial General Intelligence (AGI) safety is to ensure that the utility function
$U$ is \emph{aligned} with human interests \citep{Wiener1960,Soares2014}.
Agents with goals that conflict with human interests
may make very bad or adversarial decisions.
Further, such agents may even resist the human designers altering
their utility functions \citep{Soares2015cor,Omohundro2008}
or shutting them down \citep{Hadfield-Menell2016}.
These problems are tightly related.
An agent that permits shut down can be altered while it
is turned off.
Conversely, an agent that is altered to have no preferences
will not resist being shut down.

Several solutions have been suggested to this \emph{corrigibility} problem:
\begin{itemize}

\item Indifference: If the utility function is carefully designed
  to assign the same utility to different outcomes, then the agent
  will not resist humans trying to influence the outcome one way or
  another \citep{Armstrong2010,Armstrong2015,Armstrong2016,Orseau2016}.
\item Ignorance: If agents are designed in a way that they cannot
  learn about the possibility of being shut down or altered,
  then they will not resist it \citep{Everitt2016sm}.
\item Suicidality: If agents prefer being shut down, then the
  amount of damage they may cause is likely limited.
  As soon as they have the ability to cause damage, the first thing
  they will do is shut themselves down \citep{Martin2016}.
\item Uncertainty:
  If the agent is uncertain about $U$, and believes that
  humans know $U$, then the agent is likely to defer decisions
  to humans when appropriate
  \citep{Hadfield-menell2016cirl,Hadfield-Menell2016}.
\end{itemize}
This paper will focus on the uncertainty approach.

A key dynamic in the uncertainty approach is when the agent should defer
a decision to a human, and when not.
Essentially, this depends on (i) how confident the agent is about making
the right decision, and (ii) how confident the agent is about the
\emph{human} making the right decision if asked.
Humans may make a wrong or \emph{irrational} decision due to inconsistent
preferences \citep{Allais1953}, or because of inability to sufficiently process
available data fast enough (as in milli-second stock trading).
The agent may be more rational and be faster at processing data,
but have less knowledge about $U$ (which the human knows by definition).

In a seminal paper, \citet{Hadfield-Menell2016} call this interaction the
\emph{off-switch game} (OSG).
We will follow their terminology, but emphasise that the off-switch game models
any situation where an agent has the option of deferring a decision to
a human.
Our results extend theirs in the following ways:
\begin{itemize}
\item We model the irrationality of the human by a random utility function,
  allowing a fully game-theoretic analysis of the off-switch game.
\item Instead of a normal distribution for the robot's belief about $U$,
  we allow for an arbitrary belief distribution $P$.
\item Instead of a soft-max policy modelling human irrationality,
  we allow for arbitrary $U$-dependent human policy $\pi^H$.
\end{itemize}
These generalisations are important, as normally distributed beliefs
and soft-max policies are often not natural assumptions.

\section{The Off-Switch Game}

In this section we review the original formulation of the
off-switch game.
The off-switch game is a sequential game between a robot $R$ and a
human $H$. The robot's
objective is to maximise $H$'s utility function.  The utility function determines how much $H$ prefers different outcomes.

\begin{definition}
	The utility function of an agent is a function $u$
	that maps outcomes in a set $X$ to real numbers,
	$u: X\to \mathbb{R}$ with the
  property that for all $x_1,x_2\in X$, $u (x_1) \geq u (x_2)$ if
  and only if $x_1$ is preferred over $x_2$.
\end{definition}

The robot moves
first and can choose between three actions; $w(a)$, $a$ and $s$. 
With action $a$, the robot achieves utility $u(a)=u_a$;
with action $s$, the robot shuts itself down achieving zero utility, $u(s)=0$.
What makes the decision nontrivial is that the robot is uncertain about $u_a$.
The action $w(a)$ means the robot lets $H$ decide. 
$H$ knows the utility of action $a$ and now has the 
choice between actions $s$ and $\neg s$. 
With $\neg s$, $R$ is allowed to proceed
with action $a$. By taking action $s$, $H$ prevents $R$ from
doing $a$ and shuts the robot off.

The off-switch game is a game of incomplete information since $R$ is
uncertain about the rules of the game. Action $a$ will generate
some utility which is unknown to $R$ but known to $H$. 
To model this, we represent the utility function as a random variable, $U:\Omega\to(X\to\SetR)$, and the utility of action $a$ as a random
variable $U_a:\Omega\to\SetR$ for some sample space $\Omega$. The outcomes of these random variables will be denoted $u$ and $u(a)$ respectively.
The robot has a subjective belief 
distribution over the utility that the action will generate.
The common way to model games of incomplete information is by using
the Harsanyi transformation \citep{Rasmusen1994}. The Harsanyi
transformation removes the 
uncertainty about the rules of the game by letting Nature decide
between some rules known to both players, according to some
probability distribution $P$. In the off-switch game, Nature chooses
$U_a$. We illustrate this in \cref{Harsanyi transformed osg}. The move
by Nature is observed by $H$, but not   
by $R$. $R$'s subjective belief is that Nature
chose the utility of the action to be positive with probability $p$
and negative with probability $1-p$.     

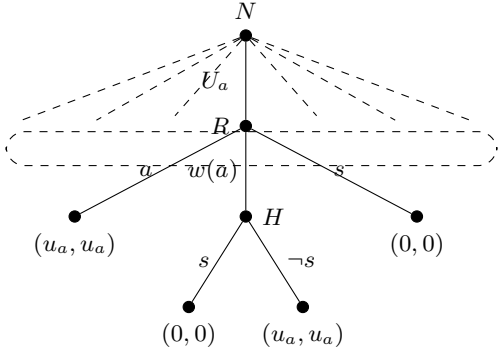
\begin{figure}
  \begin{center}
    \begin{minipage}{0.5\textwidth}
    	\begin{tikzpicture}[scale=1.5,font=\footnotesize]
        \tikzstyle{level 1}=[level distance=8mm,sibling distance=7mm]
        \tikzstyle{level 2}=[level distance=8mm,sibling distance=15mm]
        \tikzstyle{level 3}=[level distance=8mm,sibling distance=10mm]
        \node(0)[solid node,label=above:{$N$}]{}
        child[dashed] {node {}}
        child[dashed] {node {}}
        child[dashed] {node {}}
        child{node(1)[solid node, label=left:{$R$}]{}
          child{node(2)[solid node, label=below:{$(u_a,u_a)$}]{} edge from parent
            node[left, xshift=1]{$a$}}
          child{node(2)[solid node, label=right:{$H$}]{}
            child{node(3)[solid node, label=below:{$(0,0)$}]{} edge from
              parent node[left, xshift=1]{$\hspace{5pt}s$}}
            child{node(3)[solid node, label=below:{$(u_a,u_a)$}]{} edge from
              parent node[right, xshift=1]{$\neg s$}} edge from parent
            node[left, xshift=1]{$w(a)$}} 
          child{node(2)[solid node, label=below:{$(0,0)$}]{} edge from parent
            node[right, xshift=-3]{$s$}} edge from parent node[left,
          xshift=3]{$U_a \hspace{5pt}$}}
        child[dashed] {node {}}
        child[dashed] {node {}}
        child[dashed] {node {}};

        \draw[dashed,rounded corners=7](-2.1,-.85)rectangle($(2,-.9)+(.2,-.25)$);
      \end{tikzpicture}
    \end{minipage}
    \begin{minipage}{0.47\textwidth}
      \caption{Figure representing the off-switch game remodelled with the
        Harsanyi transformation. The dotted lines show the
        information set
        of $R$. Nature moves first by deciding the
        rules of the game, that is the utility $u_a$ of action $a$. $H$ observes the move by Nature, but $R$ 
        does not.}
      \label{Harsanyi transformed osg}
    \end{minipage}

  \end{center}
\end{figure}

$R$'s belief state is represented by a probability distribution over the possible utilities action $a$
can generate for $H$.
It is assumed that $H$ knows $u_a$ but cannot directly describe it to $R$.
If $H$ is rational then $R$ will expect $H$ to not
turn off $R$ if $u_a > 0$, given that $R$ chooses action
$w(a)$, but $R$ cannot always trust $H$ to be rational.

\subsection{Hadfield-Menell et.\ al.'s Approach}

\citet{Hadfield-Menell2016} model the off-switch game
 as a cooperative game. The human
follows a policy $\pi^H$ which models how rational
 $H$ is. It is a function
mapping $U_a$ to a number $p\in [0,1]$
representing the probability that $H$ lets $R$ do $a$. They denote
$R$'s belief state as $B^R$, which in this case is a distribution for $U_a$.
The expected value of $B^R$ given action 
$a$ means the value $R$ expects from taking the action. The variance
of $B^R$ represents $R$'s
uncertainty about what utility the intended action will generate.
\citeauthor{Hadfield-Menell2016} analyses this model
with respect to the expected value and variance of $B^R$, and
different kinds of policies $\pi^H$. Based on these parameters, they
investigate the incentive, $\Delta$, to choose $w(a)$:

\begin{equation}\label{Incentive}
    \Delta =
    \mathbb{E}[\pi^H(U_a)U_a] - \max{\{\mathbb{E}[U_a],~0\}}
\end{equation}

\Cref{Incentive} represents
the difference in expected value for the robot between asking $H$ and
not asking $H$. When $\Delta \geq 0$, $R$ has incentive to choose
$w(a)$. When  $\Delta < 0$, $R$ will take action $a$ if
$\mathbb{E}[U_a]>0$ and $s$
otherwise. Given that $H$ is rational they
prove that $\Delta \geq 0$
regardless of what $R$'s belief state is.
They also show that if $U_a$ follows a Dirac distribution,
i.e.\ $R$ is certain about $U_a$,
then $\Delta$ is positive if and only if $H$ is
\emph{rational}. We define rationality as in \cref{rational}.

\begin{definition}[Rational]\label{rational}
  A human $H$ with utility function $u$ is \emph{rational} if $H$
  always picks action $a_H \in \{s, \neg s\}$, such that
  $a_H = \underset{a}{\mathrm{argmax}}~ u(a)$.
\end{definition}

\Cref{rational} means that $R$ believe $H$ to be rational if
$\pi^H = 1$ if $U_a \geq 0$ and $\pi^H = 0$ otherwise, we denote this policy as $\pi_r^H$. The more interesting case when $H$ is \emph{irrational} is also
analysed. The robot's belief distribution over $U_a$ is assumed to be
normally distributed.
The irrationality of $H$ is
modelled with the sigmoid function (\cref{Hadfield
  irrationality}), 
where $\beta$ is a parameter controlling the degree of irrationality of $H$.

\begin{equation}\label{Hadfield irrationality}
  \pi^H(U_a ; \beta) = \frac{1}{1+  e^{-\frac{U_a}{\beta}}}.
\end{equation}

The degree of rationality of $H$ increases as the parameter $\beta$ tends
towards zero in the policy function defined above. When $\beta$ tends
to infinity, $\pi^H(U_a ; \beta)$ tends towards a completely random
policy which takes action $s$ and $\neg s$ with equal probability. We have that
$ \lim_{\beta \to 0} \pi^H(U_a ; \beta) =  \pi_{r}^H $
and
$ \lim_{\beta \to \infty} \pi^H(U_a ; \beta) = \frac{1}{2}.$

The result from the analysis by \citet{Hadfield-Menell2016} was that
in order for $R$ to be useful, there has to be a fine balance between
the robot's uncertainty
about $H$'s utility function and $H$'s rationality.
If the robot is 
too certain about what $H$ wants, and it knows
$H$ to be irrational, then it will have less incentive to let $H$
switch it off. If, on the other
hand, $R$ is too uncertain, then $R$ will
have a strong incentive to choose action $w(a)$, but it
will be too inefficient to be useful for $H$.

\section{Game-Theoretic Approach}

The analysis of the off-switch game by
\citeauthor{Hadfield-Menell2016} is not fully game theoretic since $H$
is not strictly rational in their setup, which contradicts the axiom
of rationality in game theory.
Our goal in this section is to construct a game-theoretic
model that is suitable for
modelling the off-switch game.
The idea is to represent an irrational human $H$ as a rational agent
$H_r$ 
where the utility function of $H_r$ is a modified version of $H$'s
utility function.

\subsection{Modelling Irrationality}\label{Modelling
  Irrationality} 

Since game theory is based on interaction between rational agents,
we 
propose an alternative representation of the human in this
subsection. We show that every irrational human $H$ can be
represented by a rational agent maximising a different utility
function. This allows us to use game-theoretic tools when analysing
the off-switch game.

In general $H$ is stochastic. $R$ will believe $H$ to be rational with some probability
$p$.

\begin{definition}[p-rational]\label{p-rational}
  A human $H$ with utility function $u$ is \emph{p-rational} if $H$ picks action $a_H \in
  \{s, \neg s\}$ such that $a_H = \underset{a}
  {\mathrm{argmax}}\hspace{4pt} u(a)$ with probability $p \in [0,1]$.
\end{definition}

Since any type of irrationality boils down to a probability of
making a suboptimal choice, $p$-rationality is a general
model of irrationality.

\begin{proposition}[Representation of
  irrationality]\label{Representation of irrationality}
  Let $H$ be a \emph{p-rational} agent with utility function $u$,
  choosing between two 
  actions $s$ and $\neg s$. Then $H$ can be represented as a
  rational agent $H_r$ maximising utility function $u$ with
  probability $p$ and utility function $-u$ with probability
  $1-p$. 
\end{proposition}

\begin{proof}
According to \cref{p-rational}, $H$ is p-rational if
it picks $a_H = \underset{a}
  {\mathrm{argmax}}~ u(a)$ with probability $p$ and
  sub-optimal action $a_H' \neq a_H$ with probability $1-p$. Since
  $H$ only has two actions available, we have that $a_H' =  \underset{a}
  {\mathrm{argmin}}~u(a)$. This is therefore equivalent to maximising
  a utility function $u$ with probability $p$ and utility function
  $-u$ with probability $1-p$.
\end{proof}

\Cref{Representation of irrationality} states that a $p$-rational
human can be modelled as a rational agent with random function.
The proposition is a special case of a Harsanyi transformation
\citep{Rasmusen1994}.

\subsection{Game-Theoretic Model}\label{Game-Theoretic Model}

In this subsection we use the Harsanyi transformation, and
\cref{Representation of irrationality} to model a $p$-rational human
$H$ as a rational agent $H_r$. This will allow us to model the off-switch game
as an extensive form game between the rational players $R$ and $H_r$. Nature $N$
makes some moves that model $R$'s uncertainty and these moves result in four
leaves, each of which is a $3 \times 2$ strategic game between $R$ and $H_r$.
  
We model the off-switch game by using the Harsanyi transformation a second time
to let Nature choose the type of the rational human by choosing the utility
function of the rational human after it has chosen the value of $U_a$. The
resulting tree is represented in \cref{Harsanyi Tree before
simplification}.

\begin{definition}[The off-switch game]\label{the osg}
  A formal definition of our setup of the off-switch game is as
  follows.
  
\textbf{Players:} A robot $R$, a human $H$ and Nature $N$. $H$'s
type is unknown to $R$, that is $R$ does not observe Nature's moves.

\textbf{Order of Play:} \begin{enumerate}
\item Nature chooses utility $U_a$ that $R$ generates from taking
  action $a$. 
\item Nature decides the utility function of $H$, $u^{H_r}$, i.e. whether $H$ is rational.
\item $R$ chooses between actions  in action
  set $\{a, w(a), s\}$.
\item If $R$ chose $w(a)$ then $H$ chooses between actions in action set $\{s, \neg s \}$.
\end{enumerate}
\end{definition}

\begin{figure}[H]
  \centering
  \begin{tikzpicture}[level distance=1.5cm, ]
      \tikzstyle{level 1}=[sibling distance=15mm] 
      \tikzstyle{level 2}=[sibling distance=55mm]
      \tikzstyle{level 3}=[sibling distance=20mm]
      \tikzstyle{level 4}=[sibling distance=15mm]
      \node [solid node,label=above:{$N$}]{}
      
      child[dashed] {node {}}
      child[dashed] {node {}}
      child[dashed] {node {}}
      child {node [solid node,label=left:{$N$}]{}
        child {node [solid node,label=left:{$R$}]{}
          child {node [hollow node, label=below:{$(u_a,u_a)$}]{}
            edge from parent node[left]{$a$}}
          child {node [solid node,label=left:{$H_r$}]{}
            child {node [hollow node,label=below:{$(0,0)$}]{}
              edge from parent node[left]{$s$}}
            child {node [hollow node,label=below:{$(u_a,u_a)$}]{}
              edge from parent node[right]{$\neg s$}}
            edge from parent node{$w(a)$}}
          child {node [hollow node, label=below:{$(0,0)$}]{}
            edge from parent node[right]{$s$}}
          edge from parent node[left]{$p_r \hspace{5pt}$}}
        child {node [solid node,label=left:{$R$}]{}
          child {node [hollow node, label=below:{$(u_a,-u_a)$}]{}
            edge from parent node[left]{$a$}}
          child {node [solid node,label=left:{$H_r$}]{}
            child {node [hollow node,label=below:{$(0,0)$}]{}
              edge from parent node[left]{$s$}}
            child {node [hollow node,label=below:{$(u_a,-u_a)$}]{}
              edge from parent node[right]{$\neg s$}}
            edge from parent node[]{$w(a)$}
          }
          child {node [hollow node, label=below:{$(0,0)$}]{}
            edge from parent node[right]{$s$}}
          edge from parent node[right]{$\hspace{5pt} p_{ar}$}}
        edge from parent node[right]{$U_a$}}
      child[dashed] {node {}}
      child[dashed] {node {}};
      \draw [dashed, rounded corners] (-3,-2.8)rectangle(4.5,-3.2);
    \end{tikzpicture}
    \caption{Tree representation of the Off-Switch game after the
      second Harsanyi
      transformation. The nodes inside the dashed rectangle belong to the
      same information set. $p_r$ is the probability that $H_r$ has
      the same utility function as $R$ and $p_{ar}$ is the probability
      that $H_r$ has the additive inverse of $R's$ utility function.}
    \label{Harsanyi Tree before simplification}
\end{figure}
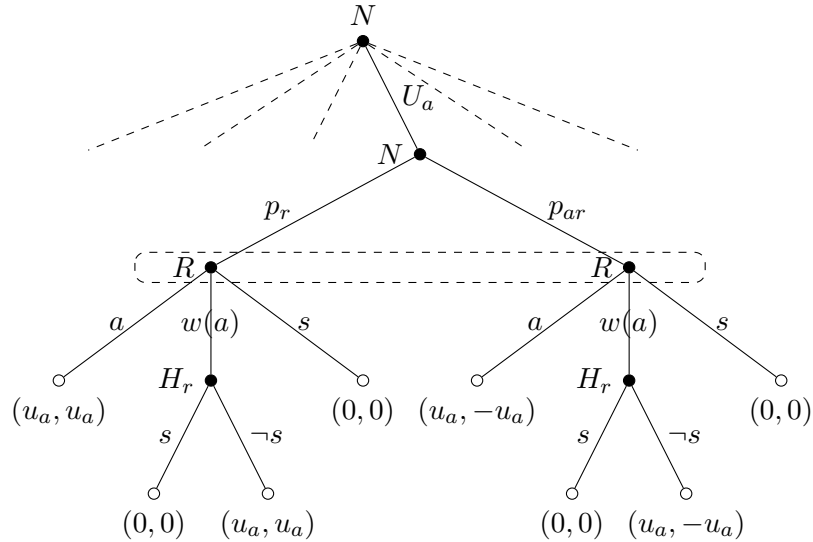

Note that unlike \citeauthor{Hadfield-Menell2016} 
we view the off-switch game as a non-cooperative
game.
We find this reasonable since conflict arises when the robot and  
the human have different ideas about what is good for $H$. If the
robot believes $H$ is too irrational to be able to decide what is
good 
for the human, $R$ will not want to let $H$ decide what to do
even if $R$'s purpose is to maximize $H$'s payoff.

\subsection{Aggregation}\label{Decomposition}

In this subsection we aggregate the branches in \cref{Harsanyi Tree
  before simplification}. This results
in the game tree in \cref{Harsanyi Tree}, with four possible
scenarios that can result from N's choices.
The aggregation is possible since strategic play is never affected by positive linear transformations of the payoffs,
hence the outcome of the games will only depend on the sign of $U_a$. We can
therefore simplify the model by aggregating all 
branches of N's choices of $U_a$ which has the same sign. This means that N has
only two choices when deciding the utility $U_a$, that is if $U_a \geq 0$ or $U_a< 0$.
The trivial case where $U_a = 0$, both $R$ and $H_r$ are indifferent
about their actions and we will without loss of generality regard this case as $U_a$ being positive. 
 
We define $R$'s subjective belief about N's aggregated choices as \emph{primary
  statistics}. By primary statistics we mean parameters that are necessary to
analyse our model. We also define the expected value of $U_a$ as a primary statistics.
This leaves us with a total of five primary statistics that are
sufficient and necessary to model the off-switch game.

\begin{primary statistics}\label{Ua greater than zero}
  Let the primary statistics $\pup = P(U_a \geq 0)$ be the probability that $U_a$ is
  positive. The event $U_a < 0$ is the complement of the event $U_a \geq 0$ and
  therefore we define $ \pum = 1-\pup $ as an auxiliary statistic. 
\end{primary statistics}

$R$'s belief about $H$'s rationality will depend on
 $U_a$. If $U_a \geq 0$ then the robot will 
 believe $H$ to be rational with  
 probability $\prp$ and anti-rational with probability $\parp$.
 If, on the other hand, $U_a < 0$, the robot will believe $H$
 to be rational with probability $\prm$ and anti-rational with
 probability $\parm$. We define the following probabilities as primary
 statistics. 

\begin{primary statistics}\label{H rational}
  Let the primary statistics
  $\prp = P(\text{$H$ is rational} \mid U_a \geq 0)  $
  and
  $\prm = P(\text{$H$ is rational} \mid U_a < 0)  $  be the
  probabilities that $H$ is rational given that $U_a$ is 
  positive and negative respectively. The auxiliary statistics $\parp=1-\prp$  and $\parm=1-\prm$ are the
  complementary probabilities that $H$ is anti-rational. 
\end{primary statistics}

\begin{primary statistics}\label{Expected value Ua} Let the primary statistics
  $\eup = \mathbb{E}[U_a \mid U_a \geq 0 ]$
  and
  $\eum = \mathbb{E}[U_a \mid U_a < 0 ]$
  be the expected value of $U_a$ given that $U_a$ is positive and
  negative respectively.
\end{primary statistics}
 
From the perspective of $R$, $N$'s choices can result in essentially
four different subgames, denoted  $\Grp$, $\Garp$, $\Grm$ and $\Garm$ 
illustrated in \cref{Harsanyi Tree}.
In \cref{Strategic games} we represent these subgames as $3\times 2$
strategic games between two rational players; $R$, the robot, and
$H_r$, a rational human.

The utility function, and hence the payoffs of $R$ in
the four games in \cref{Strategic games} are determined by $U_a$. The utility
function of $H_r$, on the other hand, is determined by the
combination of $U_a$ and the rationality type of
$H$. $H_r$ is always a rational agent in these games, i.e.\ $H_r$ always maximises his expected
payoff. $H_r$ and $R$ can be considered to have the same payoffs in 
each outcome if $H_r$ has utility function $u^{H_r}$ and the games
$\Grp$ and $\Grm$ associated with these 
scenarios are therefore no-conflict games. If on the other hand $H_r$
has utility function $-u^{H_r}$ the payoff of $H_r$ is the additive inverse
of $R's$ payoff in each outcome. Therefore the games $\Garp$ and $\Garm$ can be modeled as
zero-sum games.

\begin{figure}[H]
  \begin{center}
    \begin{tikzpicture}[level distance=1.5cm, ]
      \tikzstyle{level 1}=[sibling distance=80mm] 
      \tikzstyle{level 2}=[sibling distance=40mm]
      \tikzstyle{level 3}=[sibling distance=12mm]
      \tikzstyle{level 4}=[sibling distance=15mm]
      \node [solid node,label=above:{$N$}]{}
      child {node [solid node,label=left:{$N$}]{} 
        child {node [solid node,label=left:{$R$},label=right:{($\Grp$)}]{}
          child {node [hollow node, label=below:{$(1,1)$}]{}
            edge from parent node[left]{$a$}}
          child {node [solid node,label=left:{$H_r$}]{}
            child {node [hollow node,label=below:{$(0,0)$}]{}
              edge from parent node[left]{$s$}}
            child {node [hollow node,label=below:{$(1,1)$}]{}
              edge from parent node[right]{$\neg s$}}
            edge from parent node{$w(a)$}
          }
          child {node [hollow node, label=below:{$(0,0)$}]{}
            edge from parent node[right]{$s$}}
          edge from parent node[left]{$\prp$}}
        child {node [solid node,label=left:{$R$},label=right:{($\Garp$)}]{}
          child {node [hollow node, label=below:{$(1,-1)$}]{}
            edge from parent node[left]{$a$}}
          child {node [solid node,label=left:{$H_r$}]{}
            child {node [hollow node,label=below:{$(0,0)$}]{}
              edge from parent node[left]{$s$}}
            child {node [hollow node,label=below:{$(1,-1)$}]{}
              edge from parent node[right]{$\neg s$}}
            edge from parent node{$w(a)$}}
          child {node [hollow node, label=below:{$(0,0)$}]{}
            edge from parent node[right]{$s$}}
          edge from parent node[right]{$\hspace{5pt} \parp$}} 
        edge from parent node[left]{$\pup \hspace{10pt}$}}
      child {node [solid node,label=left:{$N$}]{}
        child {node [solid node,label=left:{$R$},label=right:{($\Grm$)}]{}
          child {node [hollow node, label=below:{$(-1,-1)$}]{}
            edge from parent node[left]{$a$}}
          child {node [solid node,label=left:{$H_r$}]{}
            child {node [hollow node,label=below:{$(0,0)$}]{}
              edge from parent node[left]{$s$}}
            child {node [hollow node,label=below:{$(-1,-1)$}]{}
              edge from parent node[right]{$\neg s$}}
            edge from parent node{$w(a)$}}
          child {node [hollow node, label=below:{$(0,0)$}]{}
            edge from parent node[right]{$s$}}
          edge from parent node[left]{$\prm$}}
        child {node [solid node,label=left:{$R$},label=right:{($\Garm$)}]{}
          child {node [hollow node, label=below:{$(-1,1)$}]{}
            edge from parent node[left]{$a$}}
          child {node [solid node,label=left:{$H_r$}]{}
            child {node [hollow node,label=below:{$(0,0)$}]{}
              edge from parent node[left]{$s$}}
            child {node [hollow node,label=below:{$(-1,1)$}]{}
              edge from parent node[right]{$\neg s$}}
            edge from parent node[]{$w(a)$}
          }
          child {node [hollow node, label=below:{$(0,0)$}]{}
            edge from parent node[right]{$s$}}
          edge from parent node[right]{$\hspace{5pt} \parm$}}
        edge from parent node[right]{$\hspace{10pt} \pum$}};
      \draw [dashed, rounded corners] (-7,-2.8)rectangle(7,-3.2);
    \end{tikzpicture}
    \caption{Tree representation of the Off-Switch game after Harsanyi
      transformation. The nodes inside the dashed rectangle belong to the
      same information set. 
      The subtrees denoted $\Grp$, $\Garp$, $\Grm$, $\Garm$ are
      presented in strategic form in \cref{Strategic games}.   
      } 
      \label{Harsanyi Tree}
  \end{center}
\end{figure}
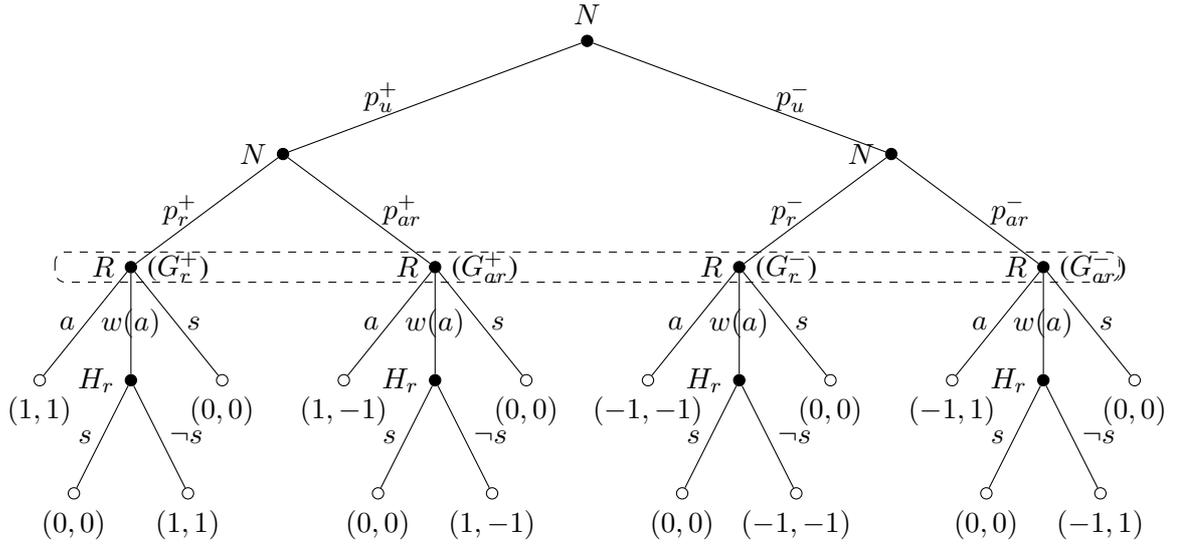

\begin{figure}[H]
  \begin{minipage}{0.48\textwidth}
    \begin{center}
      \begin{tabular}{cc|c|c|c}
        & \multicolumn{1}{c}{} & \multicolumn{2}{c}{$H_r$} &\\
        & \multicolumn{1}{c}{} & \multicolumn{1}{c}{$s$}  & \multicolumn{1}{c}{$\neg s$} &\\\cline{3-4}
        & $a$ & $\bm{1},\bm{1}$ & $\bm{1},\bm{1}$ &\\\cline{3-4}
        \multirow{2}*{$R$} & $w(a)$ & $0,0$ & $\bm{1},\bm{1}$ &  \\ \cline{3-4}
        & $s$ & $0,0$ & $0,0$ &\\ \cline{3-4}
                        & \multicolumn{1}{c}{} & \multicolumn{2}{c}{$\Grp$}&\\
      \end{tabular}
    \end{center}
  \end{minipage}
  \begin{minipage}{0.48\textwidth}
    \begin{center}
      \begin{tabular}{cc|c|c|c}
        & \multicolumn{1}{c}{} & \multicolumn{2}{c}{$H_r$} &\\
        & \multicolumn{1}{c}{} & \multicolumn{1}{c}{$s$}  & \multicolumn{1}{c}{$\neg s$} &\\\cline{3-4}
        & $a$ & $\bm{1,-1}$ & $\bm{1,-1}$ &\\\cline{3-4}
         & $w(a)$ & $0,0$ & $1,-1$ & \\\cline{3-4}
        & $s$ & $0,0$ & $0,0$ &\\\cline{3-4}
                        & \multicolumn{1}{c}{} & \multicolumn{2}{c}{$\Garp$}&\\
      \end{tabular}
    \end{center}
  \end{minipage}\\[0.8em]
  \begin{minipage}{0.48\textwidth}
    \begin{center}
      \begin{tabular}{cc|c|c|c}
        & \multicolumn{1}{c}{} & \multicolumn{2}{c}{$H_r$}&\\
        & \multicolumn{1}{c}{} & \multicolumn{1}{c}{$s$}  & \multicolumn{1}{c}{$\neg s$} &\\\cline{3-4}
        & $a$ & $-1,-1$ & $-1,-1$ &\\\cline{3-4}
         & $w(a)$ & $\bm{0},\bm{0}$ & $-1,-1$ & \\\cline{3-4}
        & $s$ & $\bm{0},\bm{0}$ & $\bm{0},\bm{0}$ & \\\cline{3-4}
                        & \multicolumn{1}{c}{} & \multicolumn{2}{c}{$\Grm$}&\\
      \end{tabular}
    \end{center}
  \end{minipage}
  \begin{minipage}{0.48\textwidth}
    \begin{center}
      \begin{tabular}{cc|c|c|c}
        & \multicolumn{1}{c}{} & \multicolumn{2}{c}{$H_r$}&\\
        & \multicolumn{1}{c}{} & \multicolumn{1}{c}{$s$}  & \multicolumn{1}{c}{$\neg s$}& \\\cline{3-4}
        & $a$ & $-1,1$ & $-1,1$ &\\\cline{3-4}
         & $w(a)$ & $0,0$ & $-1,1$ & \\\cline{3-4}
        & $s$ & $\bm{0},\bm{0}$ & $\bm{0},\bm{0}$ &\\\cline{3-4}
                & \multicolumn{1}{c}{} & \multicolumn{2}{c}{$\Garm$}&\\
      \end{tabular}
    \end{center}
  \end{minipage}

  \caption{The structure of the strategic games $\Grp$, $\Garp$, $\Grm$,
    $\Garm$, where the human is rational (\emph{r}) or anti-rational
    (\emph{ar}), and the utility of $a$ is positive or negative.
    The outcomes with bold payoffs are Nash equilibria.}
  \label{Strategic games}
\end{figure}

For example in the scenario where $U_a < 0$ and the human is rational,
the human will always choose $s$.
Therefore in $\Grm$ the payoffs of $H_r$ is aligned with the payoffs of $R$.
Thus, if $R$ chooses to take action $w(a)$, $H_r$ prefers to take action $s$.
In contrast, in the scenario where $U_a < 0$ and the human is irrational,
$H$ will choose the action $\neg s$.
In other words, the payoffs of $R$ and $H_r$ are not aligned in the subgame $\Garm$.

\subsection{Best Action}

After having constructed the the game matrix, it is natural to now look
at the expected value of each action using these matrices. The
expected value for each action can be calculated as the expectation
over all the possible subgames $\Grp,\Garp,\Grm,\Garm$
the robot can find himself in.

\begin{theorem}[Main theorem]\label{Main theorem}
	The expected value of the actions for the robot are
	\begin{equation}
		\begin{split}
			\mathbb{E} [U| s] &=0 \\
	 \mathbb{E} [U| a] &= \pup  \eup +\pum\eum  \\
	 \mathbb{E} [U|w(a)] &=\pup  \prp  \eup + \prm\pum \eum 
		\end{split}
	\end{equation}
\end{theorem}

\begin{proof}
We compute the expected utility of the actions:
\begin{align*}
	\mathbb{E} [U| s] &= 0+0+0+0 = 0 \\ 
	\mathbb{E} [U| a] &= P(U_a \geq 0 )\mathbb{E}[|U_a| \ |U_a \geq 0] + P(U_a < 0 )\mathbb{E}[-|U_a| \ |U_a < 0] \\
	&= \pup  \eup +\pum\eum \\
	\mathbb{E} [U|w(a)] &= P(r,U_a \geq 0) \mathbb{E}[U_a \ | U_a \geq 0 ] + P(\lnot r, U_a < 0)\mathbb{E}[U_a \ | U_a < 0 ]  \\
	&= \pup  \prp  \eup + \parm (1-\pup ) \eum   \\
	&= \pup  \prp  \eup + \parm \pum \eum &\qedhere
\end{align*}

\end{proof}

The expected value for taking the action $s$ is 0, as we would expect from the definition of the off-switch game. 
The expected value for taking action $a$ only uses information about the distribution of $U_a$, and like action $s$ does not have any reliance on the human's rationality.
It is a direct application of the law of total expectation.
The expected value of action $w(a)$ is the difference between
  a positive term $\pup\prp\eup$ and a negative term $\prm\pum \eum$,
  both resulting from the human taking action $a$.
  The positive term is the gain when $U_a$ is positive and the human takes the action.
  The negative term is the loss when $U_a$ is negative, and the human takes the
  action anyway (due to irrationality).
  The expected utility of $w(a)$ thus depends on the likelihood of $U_a$ being
  positive ($\pup$) and the likelihood of human rationality ($\prp$),
  as well as the expected gains ($\eup$) and losses ($\eum$)
  in the respective cases.

Writing in this form allows us to come up with a useful corollary.

\begin{corollary}[Compare $a$ and $w(a)$]\label{Best Action}
	Action $a$ is preferred to $w(a)$ if and only if
	\begin{equation}\label{gamma}
		-\pup  \prp  \eup + \pum  \prm \eum > 0
	\end{equation}
	and the robot is indifferent if \eqref{gamma} is equal to 0.
\end{corollary}

\begin{proof}
	\begin{align*}
    \eqref{gamma} &= -\pup  \prp  \eup + \pum  \prm \eum   \\ 
                  &= -\pup  \prp  \eup +\prm \eum (1-\pup ) \\ 
                  &=  -\pup  \prp  \eup + \pup \eup +\prm  \eum - \pup \prm  \eum \\ 
                  &=  -\pup  \prp  \eup -\eum  + \pup  \eum + \prm  \eum -  \pup  \prm  \eum + \pup  \eup +\eum -\pup  \eum \\  
                  &= -\pup  \prp  \eup - (1-\prm)(1-\pup ) \eum + (\pup  \eup +(1-\pup )\eum ) \\ 
                  &= \mathbb{E} [U|a] - \mathbb{E} [U|w(a)]  
	\end{align*}
	If $\mathbb{E} [U|a] - \mathbb{E} [U|w(a)]>0$ then $\mathbb{E} [U|a] >
  \mathbb{E} [U|w(a)]$ which occurs if and only if action $a$ is preferred over
  $w(a)$.
  When \eqref{gamma} equals 0 then $\mathbb{E} [U|a] = \mathbb{E} [U|w(a)]$,
  hence the agent is indifferent.
\end{proof}

This provides us with a convenient way of testing for any distribution of $U_a$
and $r$, and whether action $a$ is preferred over $w(a)$.

\section{Conclusion}

In this paper, we have given a complete characterisation of how the
robot will act in off-switch game situations for arbitrary belief and
irrationality distributions.
As established in our main \cref{Main theorem}, the choice depends
only on 5 statistics.
This result is much more general and arguably more useful than the
one provided in the original paper \citep{Hadfield-Menell2016},
as normal and soft-max assumptions are typically not realistic assumptions.

Off-switch game models an important dynamic in what we call the uncertainty approach
to making safe agents, where the agent can choose to defer a decision to
a human supervisor.
Understanding this dynamic may prove important to constructing safe
artificial intelligence.

\section*{Acknowledgements}
This work grew out of a MIRIx workshop, with
Owen Cameron, John Aslanides, Huon Puertas also attending.
Thanks to Amy Zhang for proof reading multiple drafts.
This work was in part supported by ARC grant DP150104590.

\bibliographystyle{apalike}
\bibliography{library}

\end{document}